\documentclass[9pt,shortpaper,twoside,web]{ieeecolor}
\usepackage{generic}
\usepackage{cite}
\usepackage{amsmath,amssymb,amsfonts}
\usepackage{algorithmic}
\usepackage{graphicx}
\usepackage{mathrsfs}
\usepackage{textcomp}
\def\BibTeX{{\rm B\kern-.05em{\sc i\kern-.025em b}\kern-.08em
    T\kern-.1667em\lower.7ex\hbox{E}\kern-.125emX}}
 \usepackage{epsfig}
 \usepackage[latin1]{inputenc}
\usepackage{graphicx}      
\usepackage{amsmath,amssymb}
\usepackage{pstricks,pst-node,pst-text,pst-3d,multido,pst-plot,pst-coil}
\usepackage[normalem]{ulem}
\newtheorem{remark}{Remark}
\newtheorem{proposition}{Proposition}

\newtheorem{assumption}{Assumption}
\newtheorem{lemma}{Lemma}

\def\lef[{\left[\begin{array}}
\def\rig]{\end{array}\right]}

\def\rea{\mathbb{R}}

\usepackage{psfrag}

\newcommand{\col}{ \mbox{col} }

\def\rea{\mathbb{R}}

\def\begequ{\begin{equation}}
\def\endequ{\end{equation}}
\def\lab{\label}
\def\begite{\begin{itemize}}
\def\endite{\end{itemize}}
\def\begarr{\begin{array}}
\def\endarr{\end{array}}
\def\begequarr{\begin{eqnarray}}
\def\endequarr{\end{eqnarray}}

\usepackage{multicol}

\def\calh{{\cal H}}

\def\calh{{\cal H}}

\def\liminf{\lim_{t \to \infty}}

\def\liminf{\lim_{t \to \infty}}

\def\L2{{\cal L}_2}
\def\L2e{{\cal L}_{2e}}

\def\rea{\mathbb{R}}

\def\col{\mbox{col}}

\def\et{\varepsilon_t}

\def\begmat#1{\begin{bmatrix}#1\end{bmatrix}}
\def\begali#1{\begin{align}{#1}\end{align}}
\def\begalis#1{\begin{align*}{#1}\end{align*}}

\def\begsubequ{\begin{subequations}}
\def\endsubequ{\end{subequations}}
\def\begequarr{\begin{eqnarray}}
\def\endequarr{\end{eqnarray}}
\def\begequarrs{\begin{eqnarray*}}
\def\endequarrs{\end{eqnarray*}}
\def\begarr{\begin{array}}
\def\endarr{\end{array}}
\def\begequ{\begin{equation}}
\def\endequ{\end{equation}}
\def\lab{\label}
\def\begdes{\begin{description}}
\def\enddes{\end{description}}
\def\begenu{\begin{enumerate}}
\def\begite{\begin{itemize}}
\def\endite{\end{itemize}}
\def\endenu{\end{enumerate}}

\def\lef[{\left[\begin{array}}
\def\rig]{\end{array}\right]}

\def\begcen{\begin{center}}
\def\endcen{\end{center}}
\def\begrem{\begin{remark}\rm}
\def\endrem{\end{remark}}
\def\begassums{\begin{assumption*}}
\def\endassums{\end{assumption*}}
\def\begassu{\begin{ass}}
\def\endassu{\end{ass}}
\def\beglem{\begin{lemma}}
\def\endlem{\end{lemma}}
\def\begcor{\begin{corollary}}
\def\endcor{\end{corollary}}
\def\begfac{\begin{fact}}
\def\endfac{\end{fact}}
\def\begass{\begin{assumption}}
\def\endass{\end{assumption}}

\def\begmat#1{\begin{bmatrix}#1\end{bmatrix}}
\def\begali#1{\begin{align}{#1}\end{align}}
\def\begalis#1{\begin{align*}{#1}\end{align*}}

\def\IJRNLC{{\it Int. J. on Robust and Nonlinear Control}}

\def\AUT{{\it Automatica}}

\def\CEP{{\it Control Engg. Practice}}

\newtheorem{definition}{Definition}

\usepackage{color}

\usepackage{subcaption}

\begin{document}
\title{Adaptive Compensation of Nonlinear Friction in Mechanical Systems Without Velocity Measurement}
\author{Jose Guadalupe Romero,  Romeo Ortega,  Leyan Fang and Alexey Bobtsov
\thanks{J.G. Romero, R. Ortega and Leyan Fang are with  the Department of Electrical and Electronic Engineering,  ITAM, R\'io Hondo 1, Mexico City, 01080, Mexico (e-mail: \{jose.romerovelaquez\}\{romeo.ortega\}\{leyan.fang\}@itam.mx). }
\thanks{A. Bobtsov is with the Control Systems and Robotics Department, ITMO, St. Petesburg, Russia (e-mail: bobtsov@itmo.ru).}
}
\maketitle
\begin{abstract}
Friction is an unavoidable phenomenon that exists in all mechanical systems incorporating parts with relative motion. It is well-known that friction is a serious impediment for precise servo control, hence the interest to devise a procedure to compensate for it---a subject that has been studied by many researchers for many years. The vast majority of friction compensation schemes reported in the literature rely on the availability of velocity measurements, an information that is hard to obtain. A second limitation of the existing procedures is that they rely on mathematical models of friction that contain several unknown parameters, some of them entering nonlinearly in the dynamic equations. In this paper we propose a globally convergent tracking controller for a mechanical system perturbed by static and Coulomb friction, which is a reliable mathematical model of the friction phenomenon, that does not rely one measurement of velocity. The key component is an immersion and invariance-based adaptive speed observer, used for the friction compensation. To the best of our knowledge, this is the first globally convergent solution to this challenging problem. We also present simulation results of the application of our observer for systems affected by friction, which is described by the more advanced LuGre model. 
\end{abstract}


\section{Introduction}
\lab{sec1}
%
It is well known that one of the major limitations to achieve good performance in mechanical systems is the presence of friction, which  gives rise to control problems such as static errors, limit cycles, and stick-slip. Friction is a nonlinear phenomenon difficult to describe analytically. Different models have been proposed to capture this phenomenon, ranging from simple linear {\em static} models, like stiction and Coulomb friction, to the more precise {\em dynamic} models like Dahl \cite{DAH} or the LuGre \cite{CANetal} models. A survey of these models is presented in \cite{ARMetal}, see also \cite{RUDbook} for a recent comprehensive summary of the problems of friction modeling and compensation. Model-based friction compensation requires  friction parameter estimation, hence the need of adaptive friction compensation. The problem of adaptive friction compensation has a very long history that dates at least as far back as \cite{WINGIL}. It is presented in \cite{LANbook} as an application example of model reference adaptive control. Throughout the years, many publications have been devoted to this topic \cite{CANLIS,FRIPAR,LEOKRI,PANORTGAF,SANKELSAN,YAZKHO}, all of them assuming {\em velocty is measurable}. Different friction models are considered in these publications and, with the notable exception of \cite{PANORTGAF,SANKELSAN,YAZKHO} where Lyapunov methods are used, the stability analysis is based on linearized models, hence only local.  

Very few results have been reported on adaptive friction compensation {\em without} velocity measurement. An early reference is \cite{FRIMEN} where Lyapunov's First Method is used to conclude {\em local stability} of an observer-based method. Later, in \cite{WANMELKHO}, a claim of {\em global practical stability} for an observer-based design in the presence of {\em Coulomb friction} only is made. However, a detailed analysis of the results proves that the claim is incorrect. Indeed, correcting the scaling done in between equations (39) and (40) yields a bound on the Lyapunov function which is different from (51). Actually, the corrected bound is not even sufficient to guarantee {\em boundedness} of all signals. For the case of {\em stiction compensation}, a complete answer to the problem is given in \cite{ROMORT, ROMMORMAD, ROM} using Immersion and Invariance (I\&I) \cite{ASTKARORTbook} technique to design an adaptive speed observer. Unfortunately, as it is widely recognized, stiction models are not sufficient to capture the actual effect of friction in most practical examples.

In this paper we provide a complete answer to the problem of adaptive friction compensation without velocity measurements for the classical stiction plus Coulomb friction model.\footnote{A detailed simulation analysis that validates this approximation is carried-out in the paper.} Following the I\&I methodology \cite{ASTKARORTbook} we design an adaptive observer for the unmeasurable speed---that incorporates and estimator for the unknown parameters of the friction model---for which we guarantee {\em global convergence} of the speed estimate. Equipped with this observer, we then propose a certainty equivalent-based standard position tracking regulator that, under an additional {\em excitation assumption}, needed for the convergence of the estimated parameters, is also shown to ensure global tracking of any desired time-varying reference.   

To illustrate the generality of the proposed adaptive speed observer we show in the paper the application of a slight variation of it to the more general model of an electro-hydraulic system studied in \cite{ESTetal, RAZetal}. It is shown that the I\&I-based design achieves global convergence of the speed estimate with very little prior knowledge on the system---{\em e.g.}, an upperbound on one of the friction model parameters.      

Before closing this introduction we would like to make a comment on the application of the {\em high-gain-based sliding mode} technique to the adaptive friction compensation problem. As early as 1996, even prestigious robotics researchers \cite{ARINAN} were allured by the promises of this technique. Many papers reported the use of this methodology for this problem, with the stability claims always obscured by the underlying (usually hidden) assumptions imposed by this method. As illustration of this situation we would like to comment on the very recent report \cite{ESTetal}---whose claims and assumptions are representative of the ones made in all sliding-mode papers.\footnote{See, for instance, \cite{RAZetal} and references therein.} The paper makes appeal to the latest developments on this area, namely the use of high-order sliding modes \cite{LEV} where, contrary to conventional wisdom in control theory \cite{OGAbook}, it is suggested that it is possible to {\em differentiate} the signals as many times as desired---see also \cite{CRUMOR}. In  \cite{ESTetal} this is done three times to ``obtain" an ``internal model like" representation of the system in equation (10). Several obscure signal boundedness assumptions on the system signals are then imposed, including the one that the system {\em velocity and acceleration} are bounded---although not expressed in this words it is a consequence of Assumption 1, that imposes that $\delta_2(t)$ is a ``vanishing perturbation", whose derivative  is also bounded. Assumption 2, which reads like a tautology,\footnote{It is assumed that signals inside a {\em compact set} are bounded---forgetting that boundedness is a characteristic of compact sets.}  also (apparently) assumes boundedness of the {\em whole system state}, which is in direct contradiction of the hidden assumption that the control gains can be selected as large as desired. The authors then propose to use a ``state observer", described in (13), which is a copy of the {\em linear part} of the system dynamics (10), with correction terms consisting of fractional powers of the absolute value of the poisition observation error. Although it is argued that this technique is ``model-based", this construction does not seem to reflect this claim. A final, even more distressing condition, is imposed in Assumption 3, where boundedness of the derivative of the input of the tracking error dynamics is supposed. The final claim is made in Theorem 1 where it is asserted that, under Assumptions 1-3 and a {\em ``suitable" selection} of two tuning gains, there {\em exist sufficiently large} values for two more tuning gains that ensures the tracking error goes to zero in {\em ``some" finite time}---the fuzziness of this statement can hardly be exceeded. 

The remainder of our paper is organized as follows. Section \ref{sec2} contains the main result of the paper: the proposed globally convergent adaptive speed observer for a mechanical system perturbed by stiction plus Coulomb friction forces, whose parameters are unknown. In Section \ref{sec3} a simple certainty equivalent-based globally convergent position tracking regulator, which includes a friction compensator, is proposed. In Section \ref{sec4} we consider as case study the system studied in \cite{ESTetal}. In Section \ref{sec5} we present detailed simulations illustrating the applicability of our theoretical results---even for the case when the friction exhibits a more complicated behavior. We wrap-up the paper with some concluding remarks in Section \ref{sec6}.  
\section{Adaptive Observer design}
\lab{sec2}
%
In this section we give the problem formulation and present the first main result of the paper.
\subsection{Model of the system and problem formulation}
\lab{subsec21}
%
The dynamics describing a motor actuating a load with static and Coulomb frictions are given by
\begin{eqnarray}
\dot x_1 &=&x_2 \nonumber\\
\dot x_2&=&- \theta_1 x_2 -\theta_2 \tanh(\vartheta x_2) + u,
\label{sys1}
\end{eqnarray}
where  $x_1(t) \in \rea$ and $x_2(t) \in \rea$  are the generalized position and velocity, respectively, $u(t) \in \rea$ is the control input and $\theta_1>0$, $\theta_2>0$ and $\vartheta>0$ are constant coefficients. To simplify the notation we have lumped the motor inertia, which  is assumed known, into the control signal and the constants $\theta_1$ and $\theta_2$. The assumption of known motor inertia is done without loss of generality, since its incorporation as an additional unknown parameter in the adaptive estimator proposed below is  straightforward, 

The main objective is to design a globally convergent adaptive observer for $x_2$ considering that only $x_1$ is measurable, $\vartheta$ is {\em known} and $\theta_1$ and $\theta_2$ are {\it unknown}.\footnote{The assumption that $\vartheta$ is known is reasonable because, from the practical viewpoint, it is simply taken as a ``sufficiently large" number to make the $\tanh(\cdot)$ function qualify as a suitable smooth approximation of a relay.} 
\subsection{Proposed adaptive obsesrver}
\lab{subsec22}
%
The adaptive velocity observer presented here exploits a monotonicity property, which is defined as follows.

\begin{definition} 
\label{def1}
A mapping ${\mathcal L} : \rea \to \rea$ is {\em strongly monotone} if there exists a constant $c>0$ such that
\begin{equation}
(a-b)[{\mathcal L} (a) - {\mathcal L} (b)] \geq c(a-b)^2 >0, \quad \forall a,b \in \rea, \; a\not= b.
\end{equation}
\end{definition}

  \begin{proposition}
\lab{pro1}
Consider the system \eqref{sys1} and assume $u$ is such that the state remains bounded. The I\&I adaptive velocity observer  
\begsubequ
		\lab{sys0}
		\begali{
			\lab{dotx2I}
			\dot x_{2I}=&\,-(\hat \theta_1+k_1) \hat x_2 -\hat \theta_2 \tanh (\vartheta \hat x_2) + u \\
	\lab{hx2}
	\hat x_2=& \; x_{2I} +k_1 x_1 \\	
	\lab{dotr1I}
			\dot \theta_{1I}=&\, \frac{\vartheta}{k_1} \hat x_2 (\dot x_{2I} +k_1 \hat x_2 ) \\
	\lab{hr1}
	\hat \theta_1=& \; \theta_{1I} -\frac{\vartheta}{2 k_1} \hat x^2_2 \\
	\lab{dotr2I}
			\dot \theta_{2I}=&\, \frac{\vartheta}{k_1} \tanh( \vartheta \hat x_2) (\dot x_{2I} +k_1 \hat x_2 ) \\
	\lab{hr2}
	\hat \theta_2=& \; \theta_{2I} -\frac{1}{k_1} \log(\cosh(\vartheta \hat x_2)) 
}
\endsubequ
with $k_1>0$ a {\em tuning} parameter,  ensures boundedness of all signals and
\begin{equation}
\lab{exptx2}
\lim_{t\to \infty}  [\hat x_2(t) - x_2(t)]=0
\end{equation}
for all initial conditions $(x_1(0), x_2(0),x_{2I}(0),\theta_{1I}(0),\theta_{2I}(0))$.
\end{proposition}
\begin{proof}
\label{proof1}
Let the observation and parameter estimation errors be defined as
\begequ
\lab{ex2}
\tilde x_2 := \hat x_2 -x_2, \quad \tilde \theta_1:= \hat \theta_1 -\theta_1, \quad \tilde \theta_2:= \hat \theta_2 -\theta_2.
\endequ
According to the I\&I methodology \cite{ASTKARORTbook} generate the estimates of the unknown state and parameters as the sum of a proportional and an
integral term, that is,
$$
\hat x_2= x_{2I}+x_{2P}(x_1), \quad \hat \theta_1= \theta_{1I} +\theta_{1P}(\hat x_2),  \quad \hat \theta_2= \theta_{2I} +\theta_{2P}(\hat x_2)
$$
where the mappings $x_{2P}$, $\theta_{1P}$, $\theta_{2P}$ and the dynamics of the observer states $x_{2I}$, $\theta_{1I}$ and $\theta_{2I}$ will be defined below.
\\
First, we study the dynamic behavior of  $\tilde x_2$ and compute   
\begin{align*}
\dot {\tilde x}_2 =& {\dot x}_{2I} +x'_{2P}(x_1) x_2 -{\dot x}_2  \\
=& {\dot x}_{2I} +x'_{2P}(x_1) (\hat x_2- \tilde x_2) - [ u -\theta_1 (\hat x_2- \tilde x_2) -\theta_2 \tanh (\vartheta x_2)]  \\
=& {\dot x}_{2I} +x'_{2P}(x_1) (\hat x_2- \tilde x_2) -  u +(\hat \theta_1-\tilde \theta_1) \hat x_2\\
&\;- \theta_1\tilde x_2 +\theta_2 \tanh (\vartheta x_2).
\end{align*}
Choosing the mapping  $x_{2P}(x_1)=k_1 x_1$  we obtain
\begin{align}
\dot{\tilde x}_2=&  -(k_1+\theta_1) \tilde x_2-\tilde \theta_1 \hat x_2  +\theta_2\tanh(\vartheta x_2)-\hat \theta_2 \tanh(\vartheta \hat x_2) \nonumber \\
=& -\gamma_1 \tilde x_2-\tilde \theta_1 \hat x_2  +\theta_2\tanh(\vartheta x_2)-(\tilde \theta_2 +\theta_2) \tanh(\vartheta \hat x_2) \nonumber \\
=& -\gamma_1 \tilde x_2-\tilde \theta_1 \hat x_2  -\theta_2 [\tanh(\vartheta \hat x_2)- \tanh(\vartheta x_2)] \nonumber \\
&-\tilde \theta_2 \tanh(\vartheta \hat x_2)
\lab{dex2}
\end{align}
where we used the expression of ${\dot x}_{2I}$ given in \eqref{dotx2I} to get the first identity and we defined 
\begequ
\lab{gam1}
\gamma_1:=k_1+\theta_1.
\endequ 
On the other hand, the time derivative of  $\tilde \theta_1$ is given by
\begin{eqnarray*}
\dot{\tilde \theta}_1&=& \dot \theta_{1I} +  \theta'_{1P}(\hat x_2) \dot {\hat x}_2 \nonumber \\
&=& \dot \theta_{1I} +  \theta'_{1P}(\hat x_2) (\dot x_{2I} +k_1 \dot x_1) \nonumber \\
&=& \dot \theta_{1I} +  \theta'_{1P}(\hat x_2) [\dot x_{2I} +k_1 (\hat x_2-\tilde x_2)].
\end{eqnarray*}
Now, choosing 
\begin{equation}
\label{dr1Ia}
\dot \theta_{1I}=- \theta'_{1P}(\hat x_2) [\dot x_{2I} +k_1\hat x_2],
\end{equation}
yields 
\begin{equation}
\label{dtr1}
\dot{\tilde \theta}_1=-k_1  \theta'_{1P}(\hat x_2) \tilde x_2.
\end{equation}

Finally, computing the time derivative of $\tilde \theta_2$, we get
\begin{eqnarray*}
\dot{\tilde \theta}_2&=& \dot \theta_{2I} +  \theta'_{1P}(\hat x_2) \dot {\hat x}_2 \\
&=& \dot \theta_{2I} +  \theta'_{2P}(\hat x_2) (\dot x_{2I} +k_1 \dot x_1)  \\
&=& \dot \theta_{2I} +  \theta'_{2P}(\hat x_2) [\dot x_{2I} +k_1 (\hat x_2-\tilde x_2)],
\end{eqnarray*}
with 
\begin{equation}
\label{dr2Ia}
\dot \theta_{2I}=- \theta'_{2P}(\hat x_2) [\dot x_{2I} +k_1\hat x_2],
\end{equation}
we have that 
\begin{equation}
\label{dtr2}
\dot{\tilde \theta}_2=-k_1  \theta'_{2P}(\hat x_2) \tilde x_2.
\end{equation}

We will now analyze the stability of the error dynamics  \eqref{dex2}, \eqref{dtr1} and \eqref{dtr2} with the Lyapunov function candidate
\begin{equation}
\label{lya}
{\mathcal H}(\tilde x_2, \tilde \theta_1, \tilde \theta_2)=\frac{1}{2} (\vartheta {\tilde x}^2_2 +\tilde \theta^2_1 +\tilde \theta^2_2).
\end{equation}
 Taking its time derivative we obtain  
\begin{align*}
\dot {\mathcal H}= &\; - \gamma_1\vartheta \tilde x^2_2 -\vartheta \tilde x_2 \tilde \theta_1 \hat x_2 -\vartheta \tilde x_2 \tilde \theta_2 \tanh(\vartheta \hat x_2) \\
& -\theta_2 (\vartheta \hat x_2-\vartheta x_2) [\tanh(\vartheta \hat x_2) - \tanh(\vartheta x_2)] \\
&\; -k_1 \theta'_{1P}(\hat x_2) \tilde \theta_1 \tilde x_2  -k_1 \theta'_{2P}(\hat x_2) \tilde \theta_2 \tilde x_2.
\end{align*}
Clearly, if the mappings $\theta_{1P}(\hat x_2)$  and $\theta_{2P}(\hat x_2)$ solve the ordinary differential equations (ODEs)
\begin{equation}
\label{pdes}
 \theta'_{1P}(\hat x_2)=-\frac{\vartheta}{k_1} \hat x_2, \quad  \theta'_{2P}(\hat x_2)=-\frac{\vartheta}{k_1}\tanh(\vartheta \hat x_2)
\end{equation}
one obtains 
\begin{align}
\dot {\mathcal H} &= \; - \gamma_1\vartheta \tilde x^2_2  -\theta_2 (\vartheta \hat x_2-\vartheta x_2) [\tanh(\vartheta \hat x_2) - \tanh(\vartheta x_2)] \nonumber \\
&\leq  -\vartheta^2( \gamma_1 + \theta_2 \vartheta)\tilde x^2_2,
\label{dotHf}
\end{align}
where to get the last inequality we have invoked Definition \ref{def1}. We notice that the ODEs \eqref{pdes}  are solved with 
\begin{equation}
\label{solm}
\theta_{1P}(\hat x_2)= -\frac{\vartheta}{2k_1} \hat x^2_2, \quad \theta_{2P}=-\frac{1}{k_1}\log(\cosh(\vartheta \hat x_2)).
\end{equation}
Thus, the functions \eqref{dr1Ia} and \eqref{dr2Ia} correspond to \eqref{dotr1I} and \eqref{dotr2I}, respectively.  

From \eqref{lya} and \eqref{dotHf} we conclude that $\tilde x_2 \in {\mathcal L}_2 \cap {\mathcal L}_\infty$, $\tilde \theta_1\in {\mathcal L}_\infty$ and $\tilde \theta_2\in {\mathcal L}_\infty$. Invoking \cite[Theorem 8.4]{KHAbook} we conclude that \eqref{exptx2} holds---completing the proof.
\end{proof}
\section{Adaptive Friction Compensator}
\lab{sec3}
%
In this section we are interested in proposing a position-feedback controller that, applied to system \eqref{sys1}, ensures {\em global tracking} of the reference signal $r(t) \in \rea$, whose first and second time derivative $\dot r(t),\ddot r(t) \in \rea$, respectively, are known and all of them are {\em bounded}. As is well known \cite{ASTKARORTbook,LANbook,SASBODbook}, global tracking in adaptive systems, usually requires some kind of excitation condition imposed on the reference signal. In this section we identify the excitation assumption needed to achieve global tracking of our adaptive system and state the second main result of the paper.
\subsection{Proposed tracking controller}
\lab{subsec31}
%
We aim at achieving a closed-loop dynamics of the form
\begalis{
\dot e_1 &= e_2 +\et\\
\dot e_2 &= -\alpha_1 e_1 - \alpha_2 e_2 + \et,
}
where $e_1:=x_1-r,e_2:=x_2-\dot r$, $\alpha_i,i=1,2,$ are positive constants and $\et \in \rea$ is a generic symbol for a signal {\em decaying to zero}. It is easy to see that, if $x_2$ is {\em measurable} and the friction parameters are {\em known}, the ideal control law
\begequ
\lab{uid}
u^\star=\theta_1 x_2 + \theta_2 \tanh(\vartheta x_2)+\ddot r -\alpha_1 e_1 - \alpha_2 (x_2 - \dot r),
\endequ
achieves this objective with $\et \equiv 0$. To achieve an implementable control, we propose to replace the estimated speed and parameters reported in Proposition \ref{pro1}---in a certainty-equivalent manner---in \eqref{uid}. That is, we propose the adaptive control law
\begequ
\lab{u}
u=\hat \theta_1 \hat x_2 + \hat  \theta_2 \tanh(\vartheta \hat x_2)+\ddot r -\alpha_1 e_1 - \alpha_2 (\hat x_2 - \dot r).
\endequ 
It is easy to see that replacing the estimated quantities $\hat{(\cdot)}$ by $\tilde{(\cdot)}+(\cdot)$ we obtain
$$
u=u^\star+\et,
$$
where
\begin{align}
\lab{epst}
\et:= &\theta_1 \tilde x_2 + \tilde \theta_1(x_2+\tilde x_2)+ \tilde  \theta_2 \tanh(\vartheta (x_2+\tilde x_2)) \nonumber \\
&+ \theta_2 [\tanh(\vartheta (x_2+\tilde x_2)) - \tanh(\vartheta x_2)]+\alpha_2 \tilde x_2.
\end{align}
It is clear from \eqref{epst} that if the estimated parameters $\hat \theta_i$ converge to their true values we have that $\liminf \et(t)=0$---achieving the control objective. Unfortunately, from the analysis of Proposition \ref{pro1} we cannot establish parameter convergence, without imposing some kind of {\em excitation condition}---which is articulated in the following subsection.
\subsection{Excitation condition and proof of global tracking}
\lab{subsec32}
%
It will be shown below that a sufficient condition to ensure the proposed adaptive controller is globally convergent is the following.
\begin{assumption}
\lab{ass1}
Consider the speed observer of Proposition \ref{pro1}, with the input signal given by \eqref{u}. Define the vector signal
\begequ
\lab{phi}
\phi(t):=  \left[ \begin{array}{c} \hat x_2(t) \\ \tanh(\vartheta \hat x_2(t)) \end{array} \right].
\endequ
The reference signal $r(t)$ is such that the following condition holds true. There exist sequences of positive numbers $\{t_{k}\}$, $\{T_{k}\}$, and $\{\lambda_{k}\}$ such that 
$$
t_{k+1}\ge t_{k}+T_{k}
$$
for $k=1,2,\dots$, $\inf\{T_{k}\}>0$, $\sup\{T_{k}\}<\infty$, and 
\begequ
\lab{newpe}
\int_{t_{k}}^{t_{k}+T_{k}} \phi(s)\phi^{\top}(s) ds \ge \lambda_{k} I_2
\endequ
where
\begequ
\lab{newpe1}
 \sum_{k=1}^{\infty} \lambda_{k}^{2} = \infty.
\endequ
\end{assumption}

Equipped with {\bf Assumption 1} we are in position to state the following {\em global tracking} proposition.
\begin{proposition} 
\lab{pro2}
Consider the system \eqref{sys1} in closed-loop with the control \eqref{u} where the estimated speed $\hat x_2$ and parameters $\hat \theta_i,\;i=1,2$, are generated by the adaptive speed observer of Proposition \ref{pro1}. If {\bf Assumption \ref{ass1}} holds true the {\em global tracking objective} is achieved. More precisely, we have that $\liminf \tilde \theta_i(t)=0,\;i=1,2$, consequently ensuring
$$
\liminf \begmat{e_1(t) \\ e_2(t)}={\bf 0}_{2 \times 1}
$$
where ${\bf 0}_{p \times q}$ is a $p \times q$ matrix of zeros. 
\end{proposition}

\begin{proof}
Notice that \eqref{dex2}, \eqref{dtr1} and \eqref{dtr2} together with \eqref{pdes} can be rewritten as   
\begin{align}
\dot \chi=& -\left[ \begin{array}{ccc}
\gamma_1 & \hat x_2 & \tanh(\vartheta \hat x_2) \\
-{\vartheta}\hat x_2 & 0 & 0 \\ -{\vartheta}\tanh(\vartheta \hat x_2) & 0 & 0
 \end{array} \right]\chi  \nonumber\\ 
 &-\left[ \begin{array}{c}  \theta_2 [\tanh(\vartheta \hat x_2)- \tanh(\vartheta x_2)]  \\0 \\ 0 \end{array} \right] \nonumber \\
 =& -\left[ \begin{array}{cc}
\gamma_1 & \phi^\top(t) \\
-{\vartheta}\phi(t) & {\bf 0}_{2 \times 2} 
 \end{array} \right]\chi -  \left[ \begin{array}{c}  \theta_2 \sigma (t)   \\ {\bf 0}_{2 \times 1} \end{array} \right]
 \label{Xpe}
\end{align}
with  $\chi:=\col (\tilde x_2, \tilde \theta_1, \tilde \theta_2)$,
$$
\sigma (t):= \tanh(\vartheta \hat x_2(t))- \tanh(\vartheta x_2(t))
$$
and the vector $\phi(t)$ defined in \eqref{phi}. The {\em unforced} part of the system \eqref{Xpe}---i.e., {\em with} $\sigma (t)=0$---arises in model reference adaptive control of linear time-invariant systems and it has widely been studied in the adaptive systems literature \cite{ASTKARORTbook,LANbook,SASBODbook}. In particular, it has been known for over 40 years that a necessary and sufficient condition for the linear time-varying unforced system to be {\em globally exponentially stable} is that $\phi(t)$ is {\it persistently exciting} (PE), that is
\begequ
\lab{pe}
\int^{t+T}_{t} \phi(s) \phi^\top(s) ds \geq \mu I_2
\endequ
for some $T>0$ and $\mu >0$ and all $t \geq 0$. It is well-known that PE is a {\em very restrictive} assumption, for instance, it is not satisfied for constant references. On the other hand, we note that the perturbing term $\sigma(t)$ is bounded and converges to zero if $\tilde x_2(t) \to 0$---a property that is ensured by Proposition \ref{pro1}. Consequently, to ensure $\liminf \chi(t)=0$ it is {\em sufficient} to ensure global {\em asymptotic} stability of the unforced part of the system \eqref{Xpe}. Conditions to ensure this (admittedly weaker) property have been derived in \cite{BARORT}. In particular Proposition 1 of  \cite{BARORT} {\em exactly} coincides with Proposition \ref{pro2} above. Since the proof of this result is given in  \cite{BARORT}, the proof of our claim is completed.
\end{proof}
\subsection{Discussion about Assumption \ref{ass1}}
\lab{subsec33}
%
Admittedly, the condition of {\bf Assumption \ref{ass1}} imposed on the vector $\phi(t)$ is rather cryptic. To gain some intuition in the connection between this condition and the standard PE condition \eqref{pe} let us consider instead of \eqref{newpe} and \eqref{newpe1} the more conservative conditions
\begequ
\lab{newpecon}
\lambda_{\min}\left\{\int_{t_{k}}^{t_{k+1}} \phi(t)\phi^\top (t) dt\right\} \ge \mu_{k}
\endequ
and
\begequ
\lab{nonsumcon}
\sum_{k=1}^{\infty} \frac{\mu_{k}}{1+ \|\phi(t)\|_\infty^4 \tilde T_{k}} = \infty,
\endequ
respectively, where $\|\cdot\|_\infty$ is the ${\cal L}_\infty$ norm. First, notice that  \eqref{pe} implies 
$$
\lambda_{\min}\left\{\int_{t}^{t+T} \phi(t)\phi^\top (t) dt\right\} \ge \mu,
$$
which compared with \eqref{newpecon} reveals two substantial differences. 
\begenu
\item[(i)] The integration window is not fixed (to $T$) but is now time-varying $[t_k,t_{k+1}]$.
\item[(ii)] The ``excitation level" $\mu$ is also time varying, but it should satisfy the non-summability condition \eqref{nonsumcon}.
\endenu

Additional remarks on the interpretations and implications of {\bf Assumption \ref{ass1}}, as well as a discussion on the {\em necessity} of the assumption for global asymptotic stability, may be found in \cite{BARORT}.
%
\section{Application to a Hydro-Mechanical System}
\lab{sec4}
%
To show how the adaptive speed observer proposed in Proposition \ref{pro1} can be used in scenarios different from the simple mechanical system \eqref{sys1}, in this section we develop a slight variation of this observer for the linearized hydro-mechanical system considered in \cite{ESTetal}. 

The dynamic model describing the system has the form\footnote{We have neglected the presence of an additive bounded Lipschitz function term $\delta_3(t)$, which represents a perturbation  in the  pressure load $x_3$, in \eqref{dx3h}.}
\begsubequ
		\lab{sys2}
		\begali{
			\lab{dx1h}
			\dot x_1=&\;x_2 \\
	\lab{dx2h}
	\dot x_2=& a_1 x_3 +\delta(t,x_2)\\	
	\lab{dx3h}	
	\dot x_3=& \; - a_2 x_2 - a_3 x_3 + u \\
	y = & \; x_1,
}
\endsubequ
where $x:=\col(x_1,x_2,x_3)$ with $x_1$ and $x_2$ the linear displacement and velocity of the cylinder, respectively; $x_3$ represents the differential load pressure, $\delta(t, x_2)$ is the nonlinear {\em friction force} and the positive parameters $a_i,i=1,2,3$ are assumed to be {\em known}. For the controller designed in \cite{ESTetal} no structure is given to this force and it is simply ``assumed" to be a vanishing Lipschitz function perturbation.\footnote{This assumption implies that the acceleration is ``assumed" to be bounded, as indicated in Section \ref{sec1}.}

To apply the result of Proposition \ref{pro1} we assume $\delta$ represents the sum of the stiction and Coulomb frictions, consequently
$$
\delta(x_2)=-\theta_1 x_2 -\theta_2 \tanh(\vartheta x_2).
$$ 
This leads to the following system representation 
\begali{
\nonumber
\dot x_1 &=x_2 \\
\nonumber
\dot x_2&=- \theta_1 x_2 -\theta_2 \tanh(\vartheta x_2) +a_1 x_3 \\
\dot x_3&= -a_2 x_2 - a_3 x_3 +u,
\label{sys3}
}
with $\theta_1$ and $\theta_2$ {\em unknown} and $a_i,i=1,2,3$ {\em known} positive constants.

In the proposition below we present a slight modification of the adaptive observer of Proposition \ref{pro1}, which is applicable to system \eqref{sys3}. To streamline the proposition we need the following.

\begin{assumption}
\lab{ass2} 
An upperbound on the parameter $\theta_2$ of the friction model is {\em known}.
\end{assumption}  
\begin{proposition}
\lab{pro3}
Consider system \eqref{sys3} and assume $u$ is such that the state remains bounded.  Consider the I\&I adaptive velocity observer  \eqref{sys0} with
\eqref{dotx2I} replaced by
$$
\dot x_{2I}= a_1 \hat x_3 -(\hat \theta_1 +k_1) \hat x_2- \hat \theta_2 \tanh (\vartheta \hat x_2)
$$  
with the new state equation given by
$$
	 \dot {\hat x}_3= \; -a_2 \hat x_2 -a_3\hat x_3 +u. 
$$
If {\bf Assumption \ref{ass2}} holds we can {\em compute} a positive constant $k_1^{\min}$ and {\em select} $k_1$ such that, 
\begequ
\lab{k1k1min}
k_1 \geq k_1^{\min}.
\endequ
Under these conditions we ensure boundedness of all signals and
$$
\lim_{t\to \infty}  \Big| \begmat{\hat x_2(t) - x_2(t) \\ \hat x_3(t) - x_3(t)}\Big| =0,
$$
for all initial conditions $(x_1(0), x_2(0), x_3(0),x_{2I}(0),\theta_{1I}(0),\theta_{2I}(0))$.
\end{proposition}
\begin{proof}
Let the observation and parameter estimation errors be defined as \eqref{ex2} and $\tilde x_3 := \hat x_3 -x_3$. Mimicking the proof of Proposition \ref{pro1} and after of simple calculations we have that 
\begin{align}
\dot{\tilde x}_2= & -\gamma_1 \tilde x_2-\tilde \theta_1 \hat x_2  -\theta_2 [\tanh(\vartheta \hat x_2)- \tanh(\vartheta x_2)] \nonumber \\
&-\tilde \theta_2 \tanh(\vartheta \hat x_2) + a_1 \tilde x_3
\lab{dex2a}
\end{align}
with $\gamma_1$ given in \eqref{gam1},  dynamic errors $\tilde \theta_1$ and $\tilde \theta_2$ as \eqref{dtr1} and \eqref{dtr2}, respectively, and 
\begin{equation}
\label{derx3}
\dot {\tilde x}_3 = -a_2\tilde x_2 - a_3 \tilde x_3.
\end{equation}
In this case, the stability of the error dynamics  \eqref{dex2a}, \eqref{dtr1},  \eqref{dtr2} and \eqref{derx3}  is analyzed with the Lyapunov function candidate
\begin{equation}
\label{lyaa}
{\mathcal U}(\tilde x_2, \tilde x_3, \tilde \theta_1, \tilde \theta_2):= {\mathcal H}(\tilde x_2, \tilde \theta_1, \tilde \theta_2) +\frac{1}{2}\alpha_1 \tilde x_3,
\end{equation}
with $\calh$ given in \eqref{lya} and 
\begequ
\lab{alp1}
\alpha_1>{1 \over a_1}.
\endequ
 Taking its time derivative we obtain 
\begin{align}
\dot {\mathcal U}= &\; - \gamma_1\vartheta \tilde x^2_2 -\vartheta \tilde x_2 \tilde \theta_1 \hat x_2   -\vartheta \tilde x_2 \tilde \theta_2 \tanh(\vartheta \hat x_2) +a_1 \vartheta \tilde x_3 \tilde x_2 \nonumber \\
& -\theta_2 (\vartheta \hat x_2-\vartheta x_2) [\tanh(\vartheta \hat x_2) - \tanh(\vartheta x_2)] \nonumber \\
& -k_1 \theta'_{1P}(\hat x_2) \tilde \theta_1 \tilde x_2  -k_1 \theta'_{2P}(\hat x_2) \tilde \theta_2 \tilde x_2 -\alpha_1 a_2 \tilde x_3 \tilde x_2 - \alpha_1 a_3 \tilde x^2_3. \nonumber
\end{align}
Since the mappings $\theta_{1P}(\hat x_2)$  and $\theta_{2P}(\hat x_2)$ given by \eqref{solm} solve the ODEs \eqref {pdes}
one obtains 
\begin{align}
\dot {\mathcal H} = &  - \gamma_1\vartheta \tilde x^2_2   +a_1 \vartheta \tilde x_3 \tilde x_2  -\alpha_1 a_2 \tilde x_3 \tilde x_2 - \alpha_1 a_3 \tilde x^2_3\nonumber \\
&-\theta_2 (\vartheta \hat x_2-\vartheta x_2) [\tanh(\vartheta \hat x_2) - \tanh(\vartheta x_2)] \nonumber\\
\leq &  -\vartheta( \gamma_1 + \theta_2 \vartheta) \vartheta \tilde x^2_2 +\frac{\theta_2^2 \vartheta^2}{2} \tilde x^2_2 +\frac{1}{2} \tilde x^2_3 +\frac{\alpha^2_1 a^2_2}{2} \tilde x^2_2 \nonumber \\
& + \frac{1}{2} \tilde x^2_3 - \alpha_1 a_1 \tilde x^2_3 \nonumber \\
\leq &  -\alpha_2  \tilde x^2_2   - \alpha_3 \tilde x^2_3,
\label{dotHfb}
\end{align}
with 
\begalis{
\alpha_2 &:= \vartheta  \gamma_1 + \theta_2 \vartheta^2-\frac{\theta_2^2 \vartheta^2}{2} -\frac{\alpha^2_1 a^2_2}{2}, \quad
\alpha_3 :=\alpha_1 a_1-1,
}
where, to get the first inequality, we have invoked Definition \ref{def1}. 

Condition \eqref{alp1} ensures $\alpha_3$ is positive. On the other hand, we can see that  $\gamma_1$ contains the free parameter $k_1$, which must be selected to ensure that  $\alpha_2>0$. Towards this end, we observe that knowing an upperbound on the parameter $\theta_2$ of the friction model it is possible to determine the value of $k_{\min}$ such that the condition \eqref{k1k1min} is satisfied, ensuring $\alpha_2>0$.  

From \eqref{lyaa} and \eqref{dotHfb} we conclude that $\tilde x_2, \tilde x_3 \in {\mathcal L}_2 \cap {\mathcal L}_\infty$ and $\tilde \theta_1 \in {\mathcal L}_\infty$ and $\tilde \theta_2 \in {\mathcal L}_\infty$.  Invoking \cite[Theorem 8.4]{KHAbook} we complete the proof.  

\end{proof}

\section{Simulation Results}
\lab{sec5}
%
  
In this section we test, via simulations, the {\em robustness} of our controller design, when our assumption of modeling the friction via stiction plus Coulomb forces is violated. Towards this end, we consider the following system 
\begin{eqnarray}
\dot x_1 &=&x_2, \nonumber\\
\dot x_2&=&- \sigma_2 x_2 -\sigma_1 \dot{z}-\sigma_0 z + u,\nonumber \\
\dot{z}&=&x_2-\frac{\sigma_0|x_2|z}{F_C+(F_S-F_C)e^{-\big(\frac{x_2}{v_S}\big)^2}},
\label{sys2}
\end{eqnarray}
where the friction force is represented now via the {\em LuGre model} as presented in \cite{CANetal,PANORTGAF}. The physical parameters of the LuGre model are also taken from \cite{CANetal} and are listed in Table 1. For the controller design it is {\em assumed} that the system is described by \eqref{sys1}, with $\vartheta=100$ and only $x_1$ and $u$ measurable. 

It is clear that the stiction force is represented in \eqref{sys2} by the term $\sigma_2 x_2$. On the other hand, it is argued in  \cite{ARMetal,CANetal} that the coefficient $F_C$---appearing in the dynamics of the auxiliary state $z$ in \eqref{sys2}---is ``related" with the coefficient $\theta_2$ of the Coulomb friction model of \eqref{sys1}. Therefore, with some abuse of notation, we will define the new parameter errors
\begequ
\lab{esterr}
\tilde \theta_1:=\hat \theta_1-\sigma_2,\;\tilde \theta_2:=\hat \theta_2-F_C.
\endequ

We select the desired trajectory $r(t)=\cos(w(t)*t)$ with $w(t)=0.01t$. The controller \eqref{u}, along with the adaptive observer \eqref{sys0}, are employed, with the target dynamics fixed to $\alpha_1=100$, $\alpha_2=100$---that is, the desired closed-loop poles at $p_1=-98$, $p_2=-10$. To see the effect of the tuning gain $k_1$, it is set to three different values, namely, 1, 3, and 7. The initial values of system states, observers and parameter estimates are set as follows: $x_1(0)=0.1$, $x_2(0)=0.5$, $z(0)=x_{2I}(0)=\theta_{1I}(0)=\theta_{2I}(0)=0$.

The simulation results are presented in Fig. 1 to Fig. 5. As shown in Fig. 1, despite the increasing frequency of the desired trajectory $r$, the state $x_1$ is able to track it effectively, demonstrating good reference tracking performance. It is also seen, form the zoom subfigure in Fig. 2(b) that, as expected from \eqref{gam1} and \eqref{dotHf}, the tracking error decreases with increasing values of the adaptation gain $k_1$. Fig. 2 shows that the estimated state $\hat{x}_2$ track almost perfectly the actual state $x_2$. Figures 3(a) and 3(b) show that the ``estimation errors" $\tilde{\theta}_1$ and $\tilde{\theta}_2$ defined in \eqref{esterr}, converge to zero, validating the interpretation given to the coefficient $F_C$ described above. This figure also shows the beneficial effect of increasing the gain $k_1$. Fig. 4 illustrates the trajectory of the actual control law $u$ and its deviation from the ideal control law $u^\star$ defined for the stiction plus Coulomb friction model in \eqref{uid}. Fig. 5 shows that the regressor $\phi$ satisfies the PE condition.

In the simulation above we have considered a ``rich" reference signal. To show that the controller design performs still well under severely limited excitation conditions, we have repeated the experiments above for the case of a reference signal consisting of an initial step plus a delayed ramp, as shown in Fig. 6 to Fig. 8. Despite the limited richness of the reference signal, the results in Figs. 6 and 7 demonstrate that the state $x_1$ is still able to effectively track the desired trajectory $r$, and the estimated state $\hat{x}_2$ tracks the actual state $x_2$ almost perfectly, further validating the effectiveness and robustness of the proposed approach. Moreover, Fig. 8 also illustrates the trajectory of the actual control input \( u \) and its deviation from the ideal control law \( u^\star \).

 \begin{table}[h]
		\centering
		\caption{Parameter values of LuGre model \cite{CANetal}}\label{Tab:lugparam}
		\centering
		\renewcommand {\arraystretch}{1}
		\begin{tabular}{|c|c|c|c|c|c|}
			\hline
			Parameter&Value & Parameter&Value & Parameter&Value\\ \hline 
			$\sigma_0$& $10^5$ & $\sigma_1$& $\sqrt{10^5}$ & 
			$\sigma_2$&0.4\\\hline 
			$F_C$& 1 & 
			$F_S$& 1.5 &
			$v_S$& 0.001\\
			\hline
		\end{tabular}
	\end{table}

\begin{figure}[htp]
\centering
\begin{subfigure}[b]{0.48\columnwidth}
    \includegraphics[width=\linewidth]{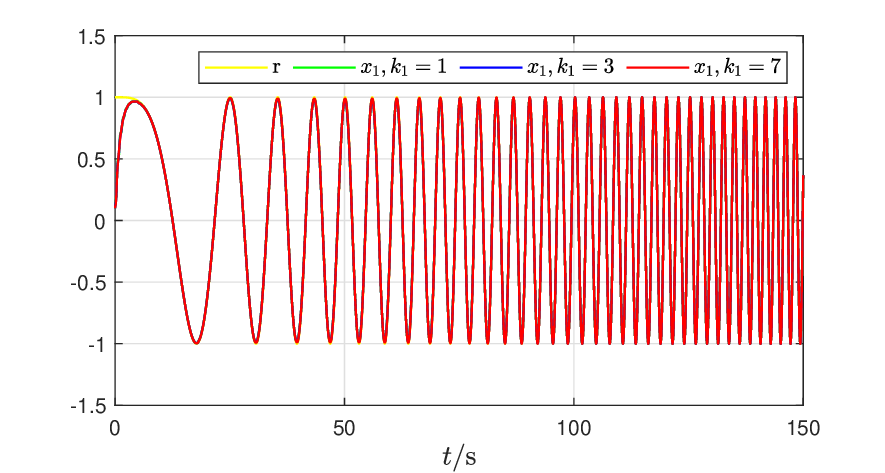}
    \caption{Comparison between $r$ and the state $x_1$}
    \label{fig1a}
\end{subfigure}
\hfill
\begin{subfigure}[b]{0.48\columnwidth}
    \includegraphics[width=\linewidth]{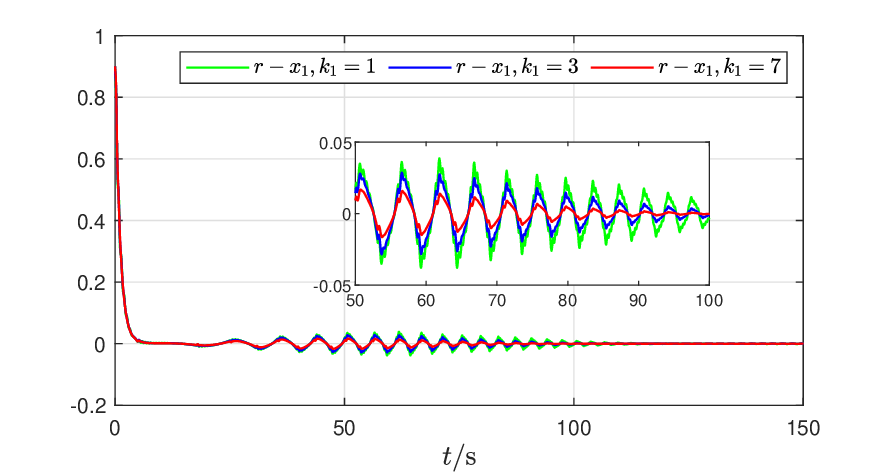}
    \caption{Trajectory of the tracking error $r - x_1$}
    \label{fig1b}
\end{subfigure}
\caption{Position tracking results under a cosine reference signal}
\label{fig1}
\end{figure}

\begin{figure}[htp]
\centering
\begin{subfigure}[b]{0.48\columnwidth}
    \includegraphics[width=\linewidth]{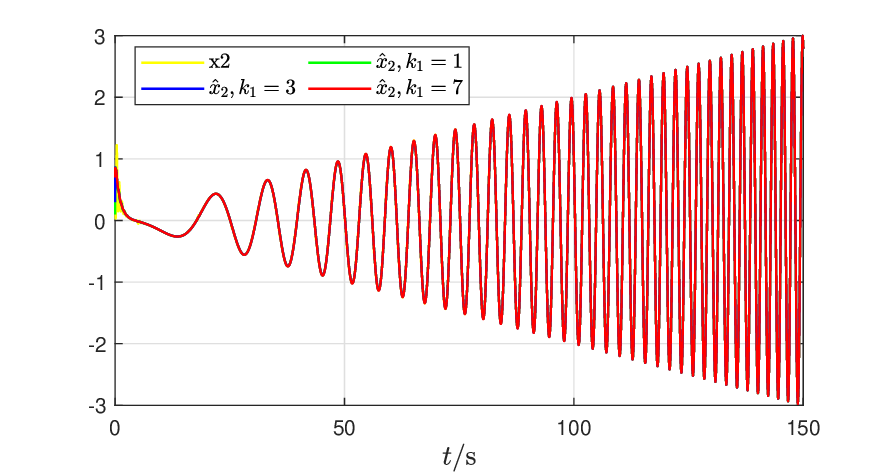}
    \caption*{(a) Comparison between $x_2$ and its estimate $\hat{x}_2$}
    \label{fig2a}
\end{subfigure}
\hfill
\begin{subfigure}[b]{0.48\columnwidth}
    \includegraphics[width=\linewidth]{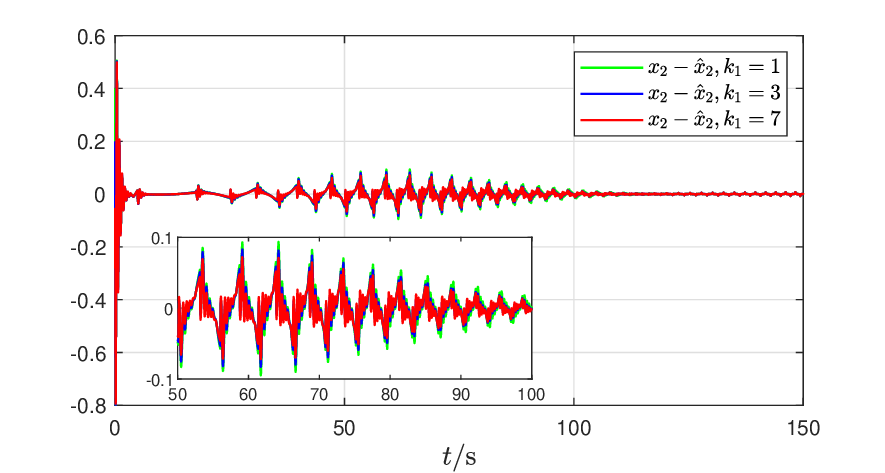}
    \caption*{(b) Trajectory of the observer error $x_2 - \hat{x}_2$}
    \label{fig2b}
\end{subfigure}
\caption{Observer results under a cosine reference signal}
\label{fig2}
\end{figure}

\begin{figure}[htp]
\centering
\begin{subfigure}[b]{0.48\columnwidth}
    \includegraphics[width=\linewidth]{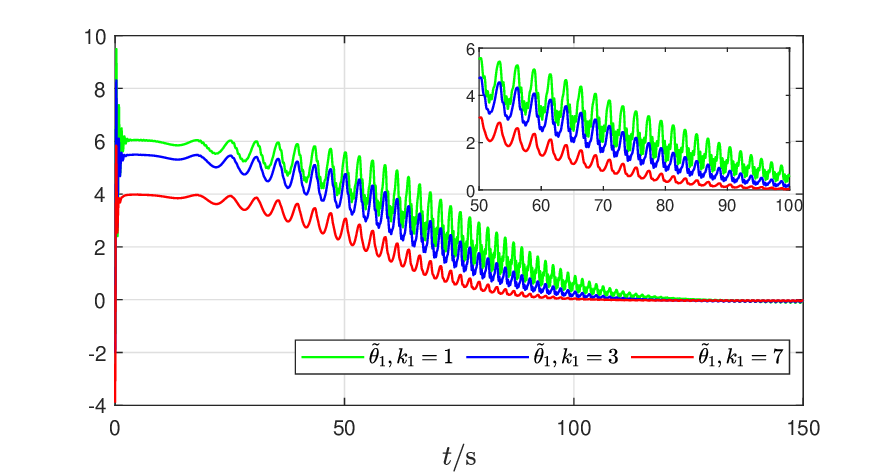}
    \caption*{(a) Trajectory of the estimation error $\tilde{\theta}_1$}
    \label{fig3a}
\end{subfigure}
\hfill
\begin{subfigure}[b]{0.48\columnwidth}
    \includegraphics[width=\linewidth]{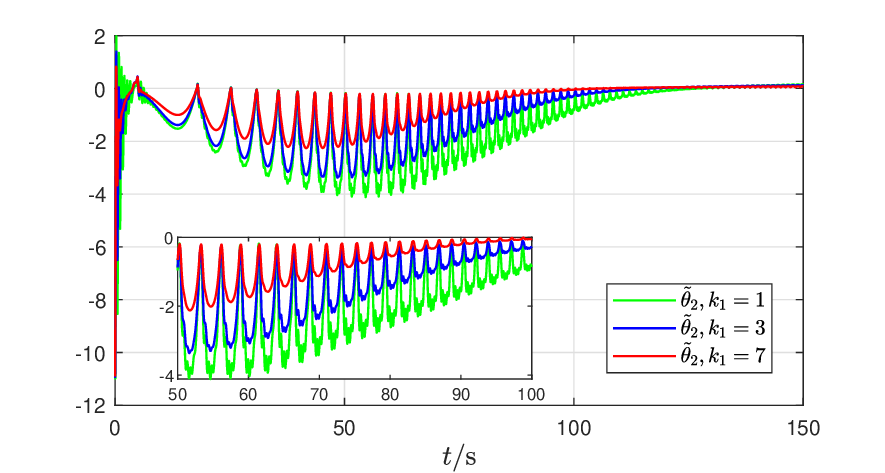}
    \caption*{(b) Trajectory of the estimation error $\tilde{\theta}_2$}
    \label{fig3b}
\end{subfigure}
\caption{Parameter estimation results under a cosine reference signal}
\label{fig3}
\end{figure}

\begin{figure}[htp]
\centering
\begin{subfigure}[b]{0.48\columnwidth}
    \includegraphics[width=\linewidth]{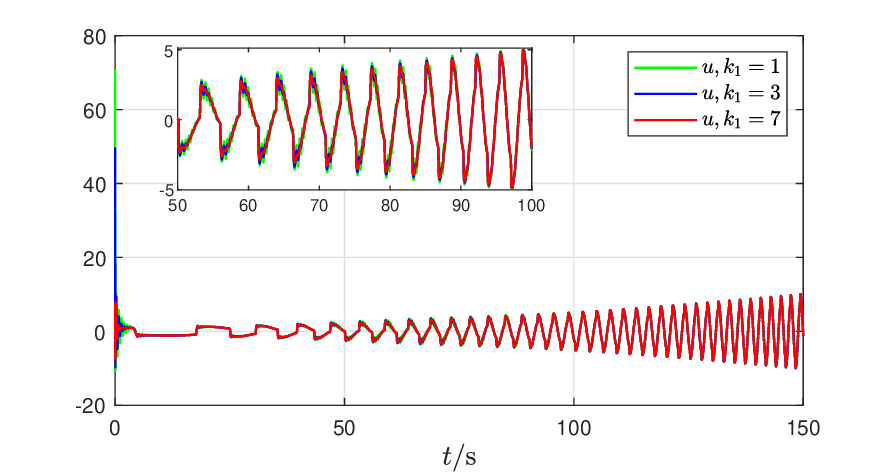}
    \caption*{(a) Trajectory of the actual control law $u$}
    \label{fig3a}
\end{subfigure}
\hfill
\begin{subfigure}[b]{0.48\columnwidth}
    \includegraphics[width=\linewidth]{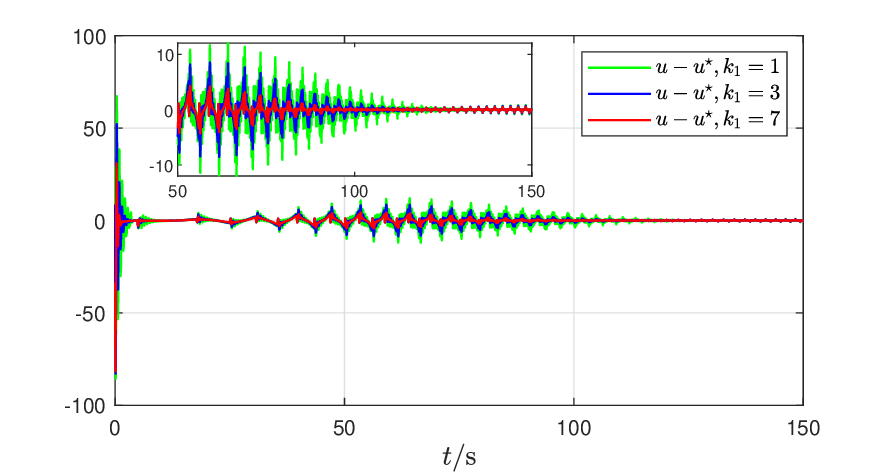}
    \caption*{(b) Deviation between $u$ and its ideal control law $u^\star$}
    \label{fig3b}
\end{subfigure}
\caption{Controller results under a cosine reference signal}
\label{fig3}
\end{figure}

\begin{figure}
\begin{centering}
\includegraphics[width=0.6\columnwidth]{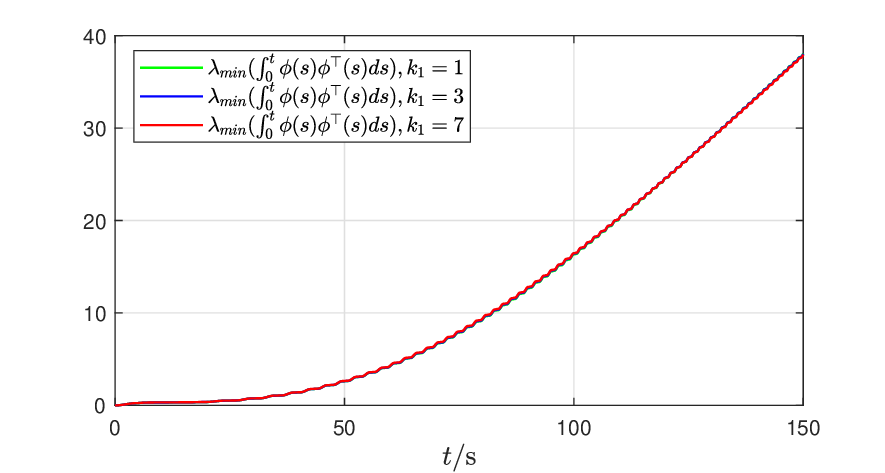}
\par\end{centering}
\caption{\label{fig4} PE condition under a cosine reference signal}
\end{figure}

\begin{figure}[htp]
\centering
\begin{subfigure}[b]{0.48\columnwidth}
    \includegraphics[width=\linewidth]{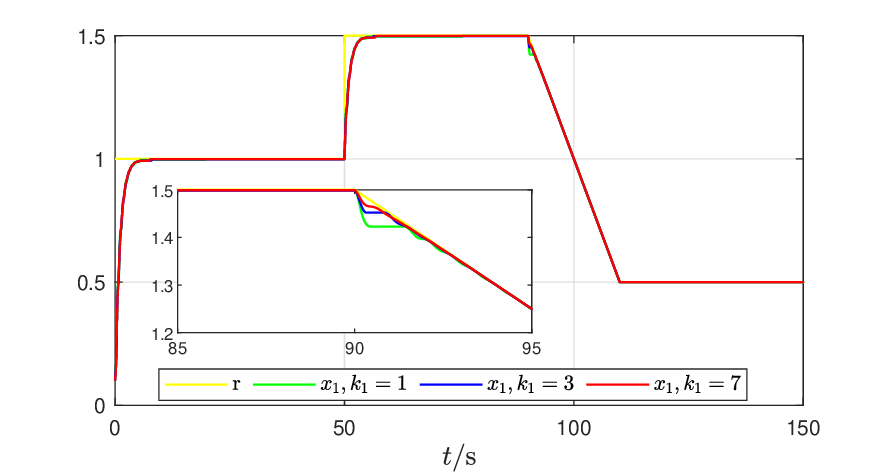}
    \caption{Comparison between $r$ and the state $x_1$}
    \label{fig1a}
\end{subfigure}
\hfill
\begin{subfigure}[b]{0.48\columnwidth}
    \includegraphics[width=\linewidth]{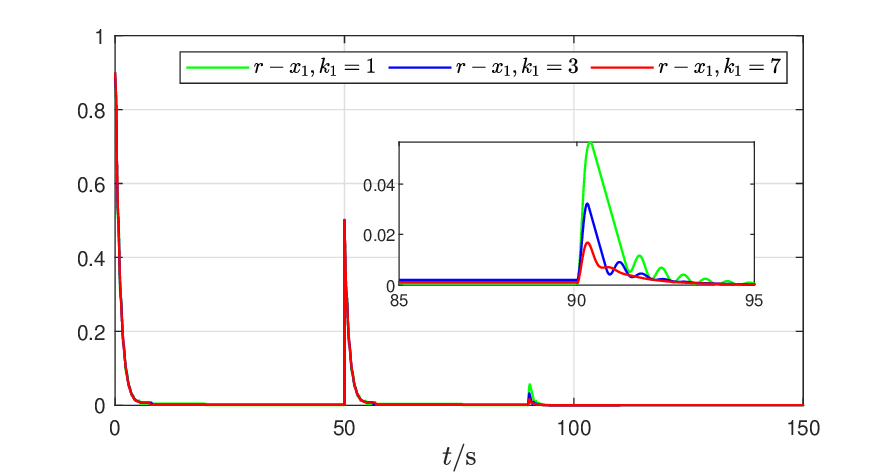}
    \caption{Trajectory of the tracking error $r - x_1$}
    \label{fig1b}
\end{subfigure}
\caption{Position tracking results under a step-plus-ramp reference signal}
\label{fig1}
\end{figure}

\begin{figure}[htp]
\centering
\begin{subfigure}[b]{0.48\columnwidth}
    \includegraphics[width=\linewidth]{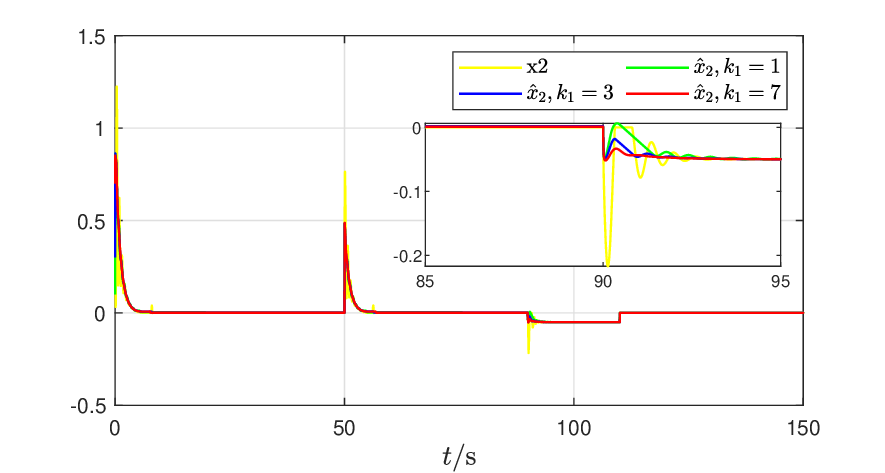}
    \caption*{(a) Comparison between $x_2$ and its estimate $\hat{x}_2$}
    \label{fig2a}
\end{subfigure}
\hfill
\begin{subfigure}[b]{0.48\columnwidth}
    \includegraphics[width=\linewidth]{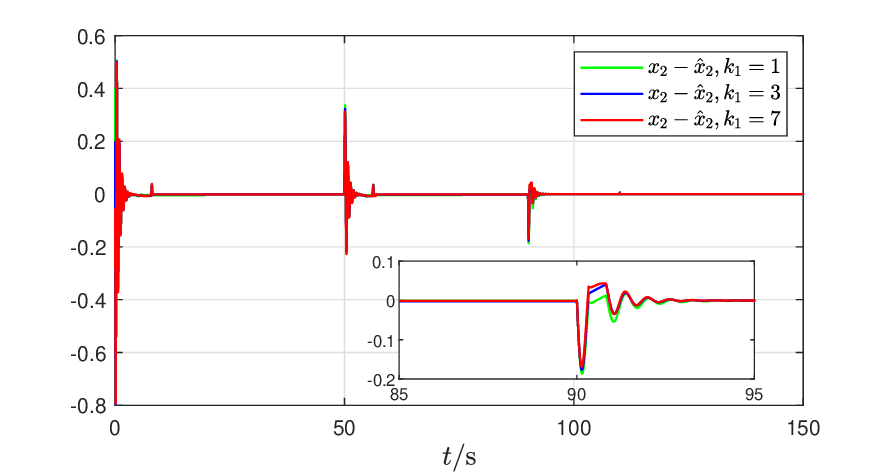}
    \caption*{(b) Trajectory of the observer error $x_2 - \hat{x}_2$}
    \label{fig2b}
\end{subfigure}
\caption{Observer results under a step-plus-ramp reference signal}
\label{fig2}
\end{figure}

\begin{figure}[htp]
\centering
\begin{subfigure}[b]{0.48\columnwidth}
    \includegraphics[width=\linewidth]{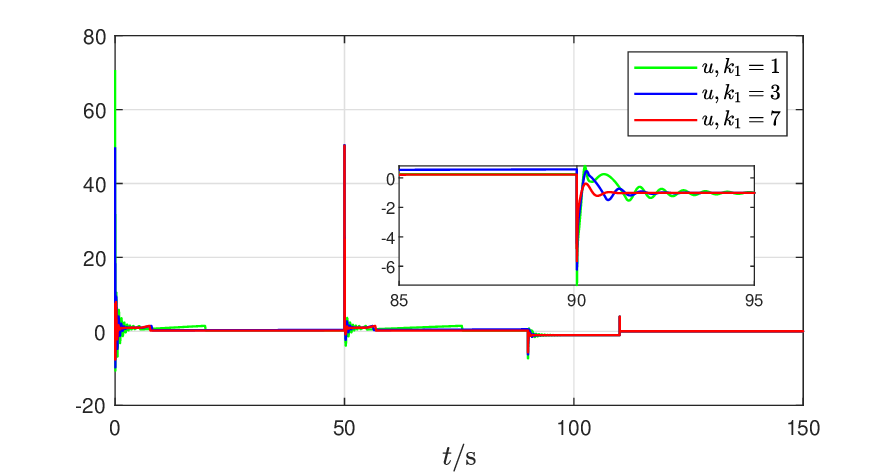}
    \caption*{(a) Trajectory of the actual control law $u$}
    \label{fig3a}
\end{subfigure}
\hfill
\begin{subfigure}[b]{0.48\columnwidth}
    \includegraphics[width=\linewidth]{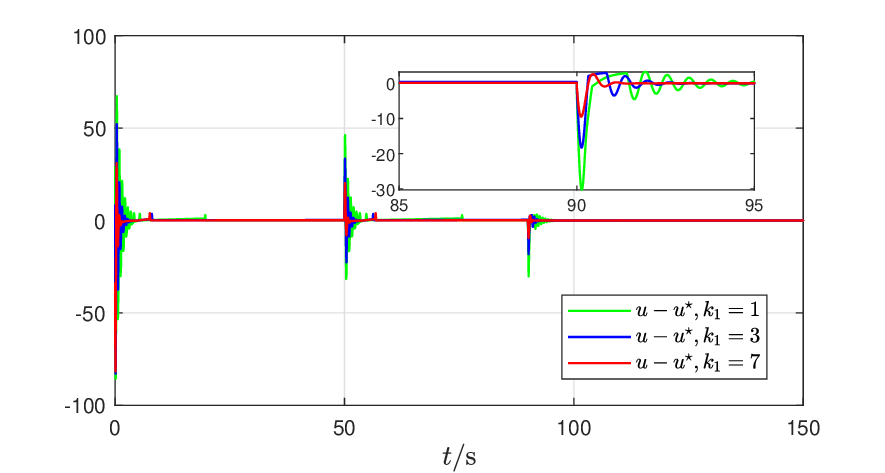}
    \caption*{(b) Deviation between $u$ and its ideal control law $u^\star$}
    \label{fig3b}
\end{subfigure}
\caption{Controller results under a step-plus-ramp reference signal}
\label{fig3}
\end{figure}
\section{Concluding Remarks}
\lab{sec6}
%
We have presented in this paper the first solution of the problem of designing a globally convergent adaptive speed observer for a simple mechanical system perturbed by friction, which is modeled by the sum of stiction and Coulomb friction terms with unknown parameters. As an outcome of this result we can design a simple---certainty equivalent-based---speed tracking controller. The observer was designed following the well-known I\&I methodology that, in contrast with sliding-mode designs, does not rely on the deleterious injection of high-gain into the control loop. Simulation results show that the procedure is also applicable to friction forces which are described by dynamic models, in particular, the widely popular LuGre model. 
 
Designing an adaptive speed observer when the noise is described by a dynamic model---like LuGre, Dahl or Stribeck models---remains a challenging open problem. The presence of products between unmeasured states and uncertain parameters hampers the application of the I\&I technique used in this paper. Its solution definitely requires the development of totally new techniques, certainly not the simplistic approach of assuming the force is a vanishing Lipschitz function perturbation adopted in the sliding mode-based literature, which is rationalized with the argument that ``positions and velocities are physically constrained".
  
As a final comment to this work we point out that, in spite of the popularity in the control community of the friction models reported in the literature, widely used commercial multibody simulation packages such as Adams, RecurDyn, and Simpack have developed their own specific stick-slip models instead of adopting one of the public domain approaches. This situation rises a question mark on the practical applicability of the mathematical models considered in the control literature. The fundamentals of these commercial models and their behavior from a practical point of view may be found in  \cite{SCHetal}.




\end{document}